\documentclass [11pt,a4paper] {article}

\usepackage[cp1252]{inputenc}

\usepackage{amssymb}
\usepackage{amsmath}
\usepackage{amsfonts,amssymb}
\usepackage[dvips]{graphicx}
\usepackage{bbm}
\usepackage{enumerate}

\usepackage{amsthm}

\DeclareMathAlphabet{\mathpzc}{OT1}{pzc}{m}{it}

\def\SmallColSep{\setlength{\arraycolsep}{1pt}}

\newtheorem{lemma}{Lemma} 
\newtheorem{theorem}{Theorem} 
\newtheorem{corollary}{Corollary} 

\begin{document}

\title{The meet of incommutable projection operators contradicts Burnside's theorem}

\author{Arkady Bolotin\footnote{$Email: arkadyv@bgu.ac.il$\vspace{5pt}} \\ \textit{Ben-Gurion University of the Negev, Beersheba (Israel)}}

\maketitle

\begin{abstract}\noindent In contrast to conjunctions of commutable projection operators unambiguously represented by their meets, the mathematical representation of conjunctions of incommutable projection operators is a question that has yet to be solved. This question relates to another asking whether the set of the column spaces of the projection operators, commutable and incommutable alike, forms a lattice. As it is demonstrated in the paper, if the Hilbert space is finite, the column spaces of the incommutable projection operators cannot be elements of one partially ordered set in accordance with Burnside’s theorem on matrix algebras.\\

\end{abstract}

\section{Introduction}  

\noindent Finding a proper mathematical representation of conjunction in the case of propositions associated with incommutable projection operators is the long-lasting problem. The essence of this problem can be presented as follows.\\

\noindent Take as an example the projection operators relating to a spin-half particle associated with the two-dimensional Hilbert space $\mathcal{H} = \mathbb{C}^2$:\smallskip

\begin{equation}  
   \hat{0}
   =
   \!\left[
      \begingroup\SmallColSep
      \begin{array}{r r}
         0 & 0 \\
         0 & 0
      \end{array}
      \endgroup
   \right]
   \,
   ,
   \,\,
   \hat{1}
   =
   \!\left[
      \begingroup\SmallColSep
      \begin{array}{r r}
         1 & 0 \\
         0 & 1
      \end{array}
      \endgroup
   \right]
   \;\;\;\;  ,
\end{equation}

\begin{equation}  
   \hat{P}_{1}^{(z)}
   =
   \!\left[
      \begingroup\SmallColSep
      \begin{array}{r r}
         1 & 0 \\
         0 & 0
      \end{array}
      \endgroup
   \right]
   \,
   ,
   \,\,
   \hat{P}_{2}^{(z)}
   =
   \!\left[
      \begingroup\SmallColSep
      \begin{array}{r r}
         0 & 0 \\
         0 & 1
      \end{array}
      \endgroup
   \right]
   \;\;\;\;  ,
\end{equation}

\begin{equation}  
   \hat{P}_{1}^{(x)}
   =
   \!\frac{1}{2}
   \!\left[
      \begingroup\SmallColSep
      \begin{array}{r r}
         1 & 1 \\
         1 & 1
      \end{array}
      \endgroup
   \right]
   \,
   ,
   \,\,
   \hat{P}_{2}^{(x)}
   =
   \!\frac{1}{2}
   \!\left[
      \begingroup\SmallColSep
      \begin{array}{r r}
         1 & -1 \\
        -1 &  1
      \end{array}
      \endgroup
   \right]
   \;\;\;\;  ,
\end{equation}

\begin{equation}  
   \hat{P}_{1}^{(y)}
   =
   \!\frac{1}{2}
   \!\left[
      \begingroup\SmallColSep
      \begin{array}{r r}
         1 & -i \\
          i &  1
      \end{array}
      \endgroup
   \right]
   \,
   ,
   \,\,
   \hat{P}_{2}^{(y)}
   =
   \!\frac{1}{2}
   \!\left[
      \begingroup\SmallColSep
      \begin{array}{r r}
         1 & i \\
        -i  & 1
      \end{array}
      \endgroup
   \right]
   \;\;\;\;  .
\end{equation}
\smallskip

\noindent These projection operators are in one-to-one correspondence with their \textit{column spaces} (a.k.a. \textit{ranges}), the closed (under addition and multiplication) subspaces of $\mathbb{C}^2$. The collection of these column spaces – denoted by $\mathcal{L}(\mathbb{C}^2)$ – is as follows\smallskip

\begin{equation}  
   \mathcal{L}(\mathbb{C}^2)
   =
   \Bigg\{
      \{0\}
      ,
      \left\{
         \!\left[
            \begingroup\SmallColSep
            \begin{array}{r}
               a \\
               0
            \end{array}
            \endgroup
         \right]\!
      \right\}
      ,
      \left\{
         \!\left[
            \begingroup\SmallColSep
            \begin{array}{r}
               0 \\
               a
            \end{array}
            \endgroup
         \right]\!
      \right\}
      ,
      \left\{
         \!\left[
            \begingroup\SmallColSep
            \begin{array}{r}
               a \\
               a
            \end{array}
            \endgroup
         \right]\!
      \right\}
      ,
      \left\{
         \!\left[
            \begingroup\SmallColSep
            \begin{array}{r}
               a \\
              -a
            \end{array}
            \endgroup
         \right]\!
      \right\}
      ,
      \left\{
         \!\left[
            \begingroup\SmallColSep
            \begin{array}{r}
               ia \\
               a
            \end{array}
            \endgroup
         \right]\!
      \right\}
      ,
      \left\{
         \!\left[
            \begingroup\SmallColSep
            \begin{array}{r}
               a \\
              ia
            \end{array}
            \endgroup
         \right]\!
      \right\}
      ,
      \mathbb{C}^2
   \Bigg\}
   \;\;\;\;  ,
\end{equation}
\smallskip

\noindent where $\{0\}$ stands for $\mathrm{ran}(\hat{0})$, the zero-dimensional subspace (containing only the zero vector), the subspace $\mathbb{C}^2$ represents $\mathrm{ran}(\hat{1})$, and $a \in \mathbb{R}$.\\

\noindent Consistent with the axioms of quantum mechanics, the partially ordered set of all the experimental propositions relating to the spin-half particle is assumed to be isomorphic to the partially ordered set of all closed subspaces of $\mathbb{C}^2$. Hence, in line with this assumption, there must be some ordering in $\mathcal{L}(\mathbb{C}^2)$, that is, there must be pairs of the subspaces, for which one subspace precedes the other, as well as pairs, in which neither subspace come first, in $\mathcal{L}(\mathbb{C}^2)$. Particularly, the partially ordered set $\mathcal{L}(\mathbb{C}^2)$ is assumed to form a lattice, namely, all pairs of the subspaces from $\mathcal{L}(\mathbb{C}^2)$ are expected to have \textit{a meet $\wedge$} represented by their intersection and \textit{a join $\vee$} represented by the smallest closed subspace of $\mathbb{C}^2$ containing their union.\\

\noindent The problem with this assumption is that \textit{it contradicts the distributive law}. To be sure, consider the following meets:\smallskip

\begin{equation} \label{one} 
   \mathrm{ran}(\hat{P}_{1}^{(z)})
   \wedge
   \mathrm{ran}(\hat{P}_{1}^{(x)})
   =
   \left\{
      \!\left[
         \begingroup\SmallColSep
         \begin{array}{r}
            a \\
            0
         \end{array}
         \endgroup
      \right]\!
   \right\}
   \cap
    \left\{
      \!\left[
         \begingroup\SmallColSep
         \begin{array}{r}
            a \\
            a
         \end{array}
         \endgroup
      \right]\!
   \right\}
   =
   \{0\}
   \;\;\;\;  .
\end{equation}

\begin{equation}  
   \mathrm{ran}(\hat{P}_{2}^{(z)})
   \wedge
   \mathrm{ran}(\hat{P}_{1}^{(x)})
   =
   \left\{
      \!\left[
         \begingroup\SmallColSep
         \begin{array}{r}
            0 \\
            a
         \end{array}
         \endgroup
      \right]\!
   \right\}
   \cap
    \left\{
      \!\left[
         \begingroup\SmallColSep
         \begin{array}{r}
            a \\
            a
         \end{array}
         \endgroup
      \right]\!
   \right\}
   =
   \{0\}
   \;\;\;\;  .
\end{equation}
\smallskip

\noindent In keeping with them, the right-hand side of the distributive axiom\smallskip

\begin{equation} \label{two} 
   \left(
      \mathrm{ran}(\hat{P}_{1}^{(z)})
      \vee
      \mathrm{ran}(\hat{P}_{2}^{(z)})
   \right)\!
   \wedge
   \mathrm{ran}(\hat{P}_{1}^{(x)})
   =
   \left(
      \mathrm{ran}(\hat{P}_{1}^{(z)})
      \wedge
      \mathrm{ran}(\hat{P}_{1}^{(x)})
   \right)\!
   \vee\!
   \left(
      \mathrm{ran}(\hat{P}_{2}^{(z)})
      \wedge
      \mathrm{ran}(\hat{P}_{1}^{(x)})
   \right)
   \;\;\;\;   
\end{equation}
\smallskip

\noindent must be the zero subspace, i.e., $\{0\} \vee \{0\} = \{0\}$. However, in view of\smallskip

\begin{equation}  
   \mathrm{ran}(\hat{P}_{1}^{(z)})
   \vee
   \mathrm{ran}(\hat{P}_{2}^{(z)})
   =
   \left(
      \left\{
         \!\left[
            \begingroup\SmallColSep
            \begin{array}{r}
               a \\
               0
            \end{array}
            \endgroup
         \right]\!
      \right\}^{\bot}\!\!
      \cap
       \left\{
         \!\left[
            \begingroup\SmallColSep
            \begin{array}{r}
               0 \\
               a
            \end{array}
            \endgroup
         \right]\!
      \right\}^{\bot}
   \right)^{\bot}\!\!
   =
   \left(
      \left\{
         \!\left[
            \begingroup\SmallColSep
            \begin{array}{r}
               0 \\
               a
            \end{array}
            \endgroup
         \right]\!
      \right\}
      \cap
       \left\{
         \!\left[
            \begingroup\SmallColSep
            \begin{array}{r}
               a \\
               0
            \end{array}
            \endgroup
         \right]\!
      \right\}
   \right)^{\bot}\!\!
   =
   \{0\}^{\bot}
   =
   \mathbb{C}^2
   \;\;\;\;  ,
\end{equation}
\smallskip

\noindent where $(\cdot)^{\bot}$ denotes the orthogonal complement of $(\cdot)$, and\smallskip

\begin{equation}  
   \mathbb{C}^2
   \cap
   \left\{
      \!\left[
         \begingroup\SmallColSep
         \begin{array}{r}
            a \\
            a
         \end{array}
         \endgroup
      \right]\!
   \right\}
   =
   \left\{
      \!\left[
         \begingroup\SmallColSep
         \begin{array}{r}
            a \\
            a
         \end{array}
         \endgroup
      \right]\!
   \right\}
   \;\;\;\;  ,
\end{equation}
\smallskip

\noindent one finds that (\ref{two}) does not hold for any $a \neq 0$:\smallskip

\begin{equation}  
    \left\{
      \!\left[
         \begingroup\SmallColSep
         \begin{array}{r}
            a \\
            a
         \end{array}
         \endgroup
      \right]\!
   \right\}
   =
    \left\{
      \!\left[
         \begingroup\SmallColSep
         \begin{array}{r}
            0 \\
            0
         \end{array}
         \endgroup
      \right]\!
   \right\}
   \;\;\;\;  .
\end{equation}
\smallskip

\noindent Thus, in order to maintain that the collection of the column spaces $\mathcal{L}(\mathbb{C}^2)$ is a lattice, \textit{it is necessary to give up distributivity}.\\

\noindent On the other hand, since\smallskip

\begin{equation}  
   \mathcal{U}
   =
   \underbrace{
   \left(
      \mathrm{ran}(\hat{P}_{1}^{(z)})
      \wedge
      \mathrm{ran}(\hat{P}_{1}^{(x)})
   \right)
   }_{
         \left\{
            \begin{smallmatrix}
               0
            \end{smallmatrix}
      \right\}
    }
   \vee
   \underbrace{
   \left(
      \mathrm{ran}(\hat{P}_{2}^{(z)})
      \wedge
      \mathrm{ran}(\hat{P}_{1}^{(x)})
   \right)
   }_{
         \left\{
            \begin{smallmatrix}
               0
            \end{smallmatrix}
      \right\}
    }
   =
   \{0\}
   \;\;\;\;  .
\end{equation}

\begin{equation}  
   \mathcal{V}
   =
   \underbrace{
   \bigg[
      \left(
         \mathrm{ran}(\hat{P}_{1}^{(z)})
         \wedge
         \mathrm{ran}(\hat{P}_{1}^{(x)})
      \right)
   \vee
   \mathrm{ran}(\hat{P}_{2}^{(z)})
   \bigg] 
   }_{
      \{0\}
      \,
      \vee   
      \,
      \mathrm{ran}(\hat{P}_{2}^{(z)})
      \,
      =
      \,
      \mathrm{ran}(\hat{P}_{2}^{(z)})
    }
   \wedge
   \,
   \mathrm{ran}(\hat{P}_{1}^{(x)})
   =
   \{0\}
   \;\;\;\;  ,
\end{equation}
\smallskip

\noindent for the mentioned above subspaces \textit{the modular identity} $\mathcal{U} = \mathcal{V}$ holds. So, it is possible to say that \textit{$\mathcal{L}(\mathbb{C}^2)$ forms an orthomodular lattice}. But if one adopts this proposal, a further problem will arise.\\

\noindent Indeed, due to the one-one correspondence between $\mathrm{ran}(\hat{P}_{i}^{(Q)})$ and $\hat{P}_{i}^{(Q)}$, where $i \in \{1,2\}$ and $Q \in \{z,x,y\}$, the meet (\ref{one}) corresponds to the meet\smallskip

\begin{equation}  
    \hat{P}_{1}^{(z)}\!
    \wedge
    \hat{P}_{1}^{(x)}
    =
    \hat{0}
   \;\;\;\;  ,
\end{equation}
\smallskip

\noindent which implies that the statement ``$S_z = +\textrm{\textonehalf}$ and $S_x = +\textrm{\textonehalf}$'' (where $S_Q$ denotes the spin along the $Q$-axis) is \textit{always false}. But this indicates that its negation, the statement ``$S_z = -\textrm{\textonehalf}$ or $S_x = -\textrm{\textonehalf}$'', must be \textit{always true}, which does not seem to make much physical sense.\\

\noindent A way to get around this problem is to try to find operations on the orthomodular lattice $\mathcal{L}(\mathbb{C}^2)$ that would coincide with the intersections of \textit{the compatible subspaces} (corresponding to the commutable projection operators such as $\hat{P}_{1}^{(z)}$ and $\hat{P}_{2}^{(z)}$) but would be different from the meets or even undefined on pairs of \textit{the incompatible subspaces} (corresponding to the incommutable projection operators such as $\hat{P}_{1}^{(z)}$ and $\hat{P}_{1}^{(x)}$) (see, for example, \cite{Pykacz15}, \cite{Pykacz17} and \cite{Griffiths} for details of this approach).\\

\noindent Another way to avoid the said problem is not to remove the distributive law but, instead, to falsify the assumption that the set of the column spaces, compatible and incompatible alike, forms a lattice.\\

\noindent This alternative way is realized in the presented paper.\\

\section{Invariant-subspace lattices}

\noindent Let $\{\hat{P}_{\diamond}\}$ be a set of projection operators (i.e., linear, self-adjoined, idempotent operators on a Hilbert space $\mathcal{H}$) associated with a set of propositions $\{\diamond\}$ (where the symbol $\diamond$ stands for any proposition, compound or simple) such that each operator $\hat{P}_{\diamond} \in \{\hat{P}_{\diamond}\}$ relates to only one proposition $\diamond \in \{\diamond\}$, without remainder.\\

\noindent Recall that the column space of the projection operator $\hat{P}_{\diamond}$ is the subset of the vectors $|\Psi\rangle \in \mathcal{H}$ that are in the image of $\hat{P}_{\diamond}$, namely,\smallskip

\begin{equation}  
   \mathrm{ran}(\hat{P}_{\diamond})
   =
   \left\{
      |\Psi\rangle \in \mathcal{H}
      :
      \;
      \hat{P}_{\diamond}|\Psi\rangle = |\Psi\rangle
   \right\}
   \;\;\;\;  .
\end{equation}
\smallskip

\noindent The orthogonal complement of any subspace $\mathrm{ran}(\hat{P}_{\diamond})$ is again a subspace. This complement is defined as \textit{the null space} (a.k.a. \textit{kernel}), $\mathrm{ker}(\hat{P}_{\diamond})$, i.e., the subset of the vectors $|\Psi\rangle \in \mathcal{H}$ that are mapped to zero by $\hat{P}_{\diamond}$, namely,\smallskip

\begin{equation}  
   \mathrm{ran}(\hat{P}_{\diamond})^{\bot}
   =
   \mathrm{ker}(\hat{P}_{\diamond})
   =
   \mathrm{ran}(\hat{1} - \hat{P}_{\diamond})
   =
   \left\{
      |\Psi\rangle \in \mathcal{H}
      :
      \;
      \hat{P}_{\diamond}|\Psi\rangle = 0
   \right\}
   \;\;\;\;  ,
\end{equation}
\smallskip

\noindent where the operation $\hat{1} - \hat{P}_{\diamond}$ is understood as negation of $\hat{P}_{\diamond}$, i.e., $\hat{1} - \hat{P}_{\diamond} = \neg\hat{P}_{\diamond}$. In this way, the projection operator $\hat{P}_{\diamond}$ breaks $\mathcal{H}$ into two orthogonal subspaces, namely,\smallskip

\begin{equation}  
   \mathrm{ran}(\hat{P}_{\diamond})
   +
   \mathrm{ran}(\neg\hat{P}_{\diamond})
   =
   \mathrm{ran}(\hat{1})
   =
   \mathcal{H}
   \;\;\;\;  ,
\end{equation}
\smallskip

\noindent where\smallskip

\begin{equation}  
   \mathrm{ran}(\hat{P}_{\diamond})
   \cap
   \mathrm{ran}(\neg\hat{P}_{\diamond})
   =
   \mathrm{ran}(\hat{0})
   =
   \{ 0\}
   \;\;\;\;  .
\end{equation}
\smallskip

\noindent Recall that \textit{a subspace $\mathcal{U} \subseteq \mathcal{H}$ is invariant under $\hat{P}_{\diamond}$} if\smallskip

\begin{equation}  
   |\Psi\rangle \in \mathcal{U}
   \implies
   \hat{P}_{\diamond}|\Psi\rangle \in \mathcal{U}
   \;\;\;\;  ,
\end{equation}
\smallskip

\noindent that is, $\hat{P}_{\diamond}(\mathcal{U})$ is contained in $\mathcal{U}$ and so\smallskip

\begin{equation}  
   \hat{P}_{\diamond}:
   \mathcal{U}
   \rightarrow
   \mathcal{U}
   \;\;\;\;  .
\end{equation}

\vspace*{2mm}

\begin{lemma}  
The column spaces of the projection operator $\hat{P}_{\diamond}$ and its negation $\neg\hat{P}_{\diamond}$ along with the column spaces of the zero and identity operators comprise the set of the subspaces invariant under $\hat{P}_{\diamond}$, namely,\smallskip

\begin{equation} \label{three} 
   \mathcal{L}(\hat{P}_{\diamond})
   =
   \left\{
      \mathrm{ran}(\hat{0})
      ,
      \mathrm{ran}(\hat{P}_{\diamond})
      ,
      \mathrm{ran}(\neg\hat{P}_{\diamond})
      ,
      \mathrm{ran}(\hat{1})
   \right\}
   \;\;\;\;  .
\end{equation}
\end{lemma}

\begin{proof}
Let $|\Psi\rangle \in \mathrm{ran}(\hat{P}_{\diamond})$. Since $\hat{P}_{\diamond}|\Psi\rangle = |\Psi\rangle$, $\hat{P}_{\diamond}|\Psi\rangle \in \mathrm{ran}(\hat{P}_{\diamond})$, and so $\hat{P}_{\diamond}: \mathrm{ran}(\hat{P}_{\diamond}) \rightarrow  \mathrm{ran}(\hat{P}_{\diamond})$. Similarly, let $|\Psi\rangle \in \mathrm{ran}(\neg\hat{P}_{\diamond})$. This means that $\hat{P}_{\diamond}|\Psi\rangle = 0$. On the other hand, $0 \in \mathrm{ran}(\neg\hat{P}_{\diamond})$, which implies $\hat{P}_{\diamond}: \mathrm{ran}(\neg\hat{P}_{\diamond}) \rightarrow  \mathrm{ran}(\neg\hat{P}_{\diamond})$. Furthermore, the space $\mathcal{H}$ itself as well as the zero subspace $\{0\}$ are \textit{the trivially invariant subspaces} for any projection operator $\hat{P}_{\diamond}$.
\end{proof}

\vspace*{2mm}

\noindent Let us identify a nonempty set $\Sigma^{(Q)} = \{\hat{P}_i^{(Q)}\}_{i=1}$ as \textit{a maximal context} if all the projection operators $\hat{P}_i^{(Q)} \in \Sigma^{(Q)}$ are orthogonal to each other and resolve to the identity operator, that is,\smallskip

\begin{equation}  
   \hat{P}_i^{(Q)}, \hat{P}_j^{(Q)} \in \Sigma^{(Q)}, i \neq j:
   \implies
   \hat{P}_i^{(Q)} \hat{P}_j^{(Q)}
   =
   \hat{P}_j^{(Q)} \hat{P}_i^{(Q)}
   =
   \hat{0}
   \;\;\;\;  ,
\end{equation}

\begin{equation}  
   \sum_{i=1} \hat{P}_i^{(Q)}
   =
   \hat{1}
   \;\;\;\;  .
\end{equation}
\smallskip

\noindent Let \textit{the set of the invariant subspaces for the maximal context $\mathcal{L}(\Sigma^{(Q)})$} be defined as the set of the subspaces invariant under each $\hat{P}_i^{(Q)} \in \Sigma^{(Q)}$, i.e.,\smallskip

\begin{equation}  
   \mathcal{L}(\Sigma^{(Q)})
   =
   \bigcap_{i=1}
      \mathcal{L}(\hat{P}_i^{(Q)})
   \;\;\;\;  .
\end{equation}
\smallskip

\noindent Consider $\mathcal{L}(\hat{P}_i^{(Q)}) \cap \mathcal{L}(\hat{P}_j^{(Q)})$: In keeping with (\ref{three}), it must be the intersection as follows:\smallskip

\begin{equation} \label{four} 
   \mathcal{L}(\hat{P}_i^{(Q)})
   \cap
   \mathcal{L}(\hat{P}_j^{(Q)})
   =
   \mathrm{ran}(\hat{R}_i)
   \cap
   \mathrm{ran}(\hat{R}_j)
   =
   \mathrm{ran}(\hat{R}_i \hat{R}_j)
   \;\;\;\;  ,
\end{equation}
\smallskip

\noindent where $\hat{R}_i = \{\hat{0}, \hat{P}_i^{(Q)}, \neg\hat{P}_i^{(Q)}, \hat{1}\}$ and $\hat{R}_j = \{\hat{0}, \hat{P}_j^{(Q)}, \neg\hat{P}_j^{(Q)}, \hat{1}\}$. Accordingly, the set $\mathcal{L}(\Sigma^{(Q)})$ can be written down as\smallskip

\begin{equation} \label{five} 
   \mathcal{L}(\Sigma^{(Q)})
   =
   \mathrm{ran}\left(\prod_{i=1} \hat{R_i} \right)
   =
   \left\{
      \{0\}
      ,
      \dots
      ,
      \mathrm{ran}(\hat{P}_i^{(Q)})
      ,
      \dots
      ,
      \mathrm{ran}\left(\sum_{i}^{\alpha}\hat{P}_i^{(Q)}\right)
      ,
      \dots
      ,
      \mathcal{H}
   \right\}
   \;\;\;\;  ,
\end{equation}
\smallskip

\noindent where $\alpha \in \{1,\dim(\Sigma^{(Q)})-1\}$.

\vspace*{3mm}

\begin{theorem} 
The set of the invariant subspaces invariant under each projection operator from the maximal context on the finite Hilbert space $\mathcal{H}$ forms a lattice (called invariant-subspace lattice of the maximal context).
\end{theorem} 

\begin{proof}
To prove that the set $\mathcal{L}(\Sigma^{(Q)})$ is \textit{a meet-semilattice}, it must be shown that any two elements from the set $\mathcal{L}(\Sigma^{(Q)})$ has a meet which is an element from $\mathcal{L}(\Sigma^{(Q)})$, i.e.,\noindent

\begin{equation}  
   \mathcal{U},\mathcal{V} \in \mathcal{L}(\Sigma^{(Q)})
   \implies
   \mathcal{U} \wedge \mathcal{V} \in \mathcal{L}(\Sigma^{(Q)})
   \;\;\;\;  .
\end{equation}
\smallskip

\noindent This derives directly from the follows\smallskip

\begin{equation}  
   \{0\}
   \cap
   \mathrm{ran}
   (
      \sum_j^\beta \hat{P}_j^{(Q)}
   )
   =
   \{0\}
   \in \mathcal{L}(\Sigma^{(Q)})
   \;\;\;\;  ,
\end{equation}

\begin{equation}  
   \mathrm{ran}
   (
      \sum_i^\alpha \hat{P}_i^{(Q)}
   )
   \cap
   \{0\}
   =
   \{0\}
   \in \mathcal{L}(\Sigma^{(Q)})
   \;\;\;\;  ,
\end{equation}

\begin{equation}  
   \mathrm{ran}
   (
      \sum_i^\alpha \hat{P}_i^{(Q)}
   )
   \cap
   \mathrm{ran}
   (
      \sum_j^\beta \hat{P}_j^{(Q)}
   )
   =
   \left\{
      \begin{array}{r}
          \mathrm{ran}(\sum_k^\gamma \hat{P}_k^{(Q)})
          \in
          \mathcal{L}(\Sigma^{(Q)}),
          \,\,
          i = j
          \\
          \{0\}
          \in
          \mathcal{L}(\Sigma^{(Q)}),
          \,\,
          i \neq j
      \end{array}
   \right.
   \;\;\;\;  ,
\end{equation}

\begin{equation}  
   \mathrm{ran}
   (
      \sum_i^\alpha \hat{P}_i^{(Q)}
   )
   \cap
   \mathrm{ran}(\hat{1})
   =
   \mathrm{ran}
   (
      \sum_i^\alpha \hat{P}_i^{(Q)}
   )
   \in \mathcal{L}(\Sigma^{(Q)})
   \;\;\;\;  ,
\end{equation}

\begin{equation}  
   \mathrm{ran}(\hat{1})
   \cap
   \mathrm{ran}
   (
      \sum_i^\beta \hat{P}_j^{(Q)}
   )
   =
   \mathrm{ran}
   (
     \sum_i^\beta \hat{P}_j^{(Q)}
   )
   \in \mathcal{L}(\Sigma^{(Q)})
   \;\;\;\;  ,
\end{equation}
\smallskip

\noindent where $\beta,\gamma \in \{1,\dim(\Sigma^{(Q)})-1\}$.\\

\noindent To prove that the set $\mathcal{L}(\Sigma^{(Q)})$ is \textit{a join-semilattice}, let us show that any two subspaces from $\mathcal{L}(\Sigma^{(Q)})$ has a join which is an element from $\mathcal{L}(\Sigma^{(Q)})$, namely,\smallskip

\begin{equation}  
   \mathcal{U},\mathcal{V} \in \mathcal{L}(\Sigma^{(Q)})
   \implies
   \mathcal{U} \vee \mathcal{V} \in \mathcal{L}(\Sigma^{(Q)})
   \;\;\;\;  .
\end{equation}
\smallskip

\noindent It is clear that\smallskip

\begin{equation}  
   \{0\}
   \vee
   \mathrm{ran}
   (
      \sum_j^\beta \hat{P}_j^{(Q)}
   )
   =
   \mathrm{ran}
   (
      \sum_j^\beta \hat{P}_j^{(Q)}
   )
   \in \mathcal{L}(\Sigma^{(Q)})
   \;\;\;\;  ,
\end{equation}

\begin{equation}  
   \mathrm{ran}
   (
      \sum_i^\alpha \hat{P}_i^{(Q)}
   )
   \vee
   \{0\}
   =
   \mathrm{ran}
   (
      \sum_i^\alpha \hat{P}_i^{(Q)}
   )
   \in \mathcal{L}(\Sigma^{(Q)})
   \;\;\;\;  ,
\end{equation}

\begin{equation}  
   \mathrm{ran}
   (
      \sum_i^\alpha \hat{P}_i^{(Q)}
   )
   \vee
   \mathrm{ran}(\hat{1})
   =
   \mathrm{ran}(\hat{1})
   \in \mathcal{L}(\Sigma^{(Q)})
   \;\;\;\;  ,
\end{equation}

\begin{equation}  
   \mathrm{ran}(\hat{1})
   \vee
   \mathrm{ran}
   (
      \sum_i^\beta \hat{P}_j^{(Q)}
   )
   =
   \mathrm{ran}(\hat{1})
   \in \mathcal{L}(\Sigma^{(Q)})
   \;\;\;\;  ,
\end{equation}
\smallskip

\noindent Consider $\mathrm{ran}(\sum_i^\alpha \hat{P}_i^{(Q)}) \vee \mathrm{ran}(\sum_j^\beta \hat{P}_j^{(Q)})$. Because the Hilbert space $\mathcal{H}$ is finite, one can write\smallskip

\begin{equation}  
   \mathrm{ran}(\sum_i^\alpha \hat{P}_i^{(Q)})
   \vee
   \mathrm{ran}(\sum_j^\beta \hat{P}_j^{(Q)})
   =
   \left[
      \mathrm{ran}(\sum_k^{N -\alpha} \hat{P}_k^{(Q)})
      \cap
      \mathrm{ran}(\sum_l^{N -\beta} \hat{P}_l^{(Q)})
   \right]^{\bot}
   \;\;\;\;  ,
\end{equation}
\smallskip

\noindent where $N = \dim(\Sigma^{(Q)})$. In accordance with (\ref{four}), it follows\smallskip

\begin{equation}  
   \left[
      \mathrm{ran}(\sum_k^{N -\alpha} \hat{P}_k^{(Q)})
      \cap
      \mathrm{ran}(\sum_l^{N -\beta} \hat{P}_l^{(Q)})
   \right]^{\bot}
   =
   \mathrm{ran}
      \left(
         \sum_i^N \hat{P}_i^{(Q)}
         -
         \sum_k^{N -\alpha} \hat{P}_k^{(Q)}
         \sum_l^{N -\beta} \hat{P}_l^{(Q)}
      \right)
   \in \mathcal{L}(\Sigma^{(Q)})
   \;\;\;\;  ,
\end{equation}
\smallskip

\noindent and so, the set $\mathcal{L}(\Sigma^{(Q)})$ is both a meet- and a join-semilattice.
\end{proof}

\vspace*{2mm}

\begin{theorem} 
The invariant-subspace lattice of the maximal context on the finite-dimensional Hilbert space is distributive.
\end{theorem} 

\begin{proof}
Observe in (\ref{five}) that every three elements of the lattice $\mathcal{L}(\Sigma^{(Q)})$ can be presented as $\mathrm{ran}(\hat{S})$, $\mathrm{ran}(\hat{T})$ and $\mathrm{ran}(\hat{W})$ where the following denotements are used: $\hat{S} = \{\hat{0}, \sum_{i}^\alpha \!\!\hat{P}_i^{(Q)}\!, \hat{1}\}$, $\hat{T} = \{\hat{0}, \sum_{j}^\beta \!\!\hat{P}_j^{(Q)}\!, \hat{1}\}$ and $\hat{W} = \{\hat{0}, \sum_{k}^\gamma \!\!\hat{P}_k^{(Q)}\!, \hat{1}\}$. Consider the equality $\mathcal{U} = \mathcal{V}$ where $\mathcal{U}$ and $\mathcal{V}$ stand for\smallskip

\begin{equation}  
   \mathcal{U}
   =
   \left(
      \mathrm{ran}(\hat{S})
      \wedge
      \mathrm{ran}(\hat{T})
   \right)
   \vee
   \left(
      \mathrm{ran}(\hat{S})
      \wedge
      \mathrm{ran}(\hat{W})
   \right)
   \;\;\;\;  ,
\end{equation}

\begin{equation}  
   \mathcal{V}
   =
   \mathrm{ran}(\hat{S})
   \wedge
   \left(
      \mathrm{ran}(\hat{T})
      \vee
      \mathrm{ran}(\hat{W})
   \right)
   \;\;\;\;  .
\end{equation}
\smallskip

\noindent According to (\ref{four}), $\mathcal{U} = \mathrm{ran}(\hat{S}\hat{T}) \vee \mathrm{ran}(\hat{S}\hat{W})$. As $\mathcal{H}$ is finite, one can write\smallskip

\begin{equation}  
   \mathcal{U}
   =
   \left(
      \mathrm{ran}(\hat{1} - \hat{S}\hat{T})
      \cap
      \mathrm{ran}(\hat{1} - \hat{S}\hat{W})
   \right)^{\bot}\!
   =
   \mathrm{ran}(\hat{S}\hat{T} + \hat{S}\hat{W})
   \;\;\;\;  .
\end{equation}
\smallskip

\noindent In the same way, $\mathcal{V}=\mathrm{ran}(\hat{S}) \cap \mathrm{ran}(\hat{T}+\hat{W}) = \mathrm{ran}(\hat{S}\hat{T} + \hat{S}\hat{W})$. Thus, the equality $\mathcal{U} = \mathcal{V}$ holds, which means that the lattice $\mathcal{L}(\Sigma^{(Q)})$ satisfies the distributivity axiom.
\end{proof}

\vspace*{2mm}

\begin{corollary} 
The set of projection operators generated by the maximal context on the finite-dimensional Hilbert space, namely,\smallskip

\begin{equation}  
   L(\Sigma^{(Q)})
   =
   \left\{
      \hat{0}
      ,
      \dots
      ,
      \sum_i^\alpha \hat{P}_i^{(Q)}
      ,
      \dots
      ,
      \hat{1}
   \right\}
   \;\;\;\;  ,
\end{equation}
\smallskip

\noindent forms a distributive lattice in which the meets of the projection operators are defined by\smallskip

\begin{equation} \label{six} 
   \mathrm{ran}(\hat{P}_i^{(Q)})
   \wedge
   \mathrm{ran}(\hat{P}_j^{(Q)})
   =
   \{0\}
   \,
   \implies
   \,
   \hat{P}_i^{(Q)}
   \wedge
   \hat{P}_j^{(Q)}
   =
   \hat{0}
   \;\;\;\;  .
\end{equation}
\end{corollary}

\vspace*{2mm}

\noindent This meet implies that any pair of the propositions $Q_i$ and $Q_j$ associated with the different projection operators $\hat{P}_i^{(Q)}$ and $\hat{P}_j^{(Q)}$ from the maximal context cannot be true together, i.e., $Q_i \wedge Q_j = \bot$ (where $\bot$ denotes \textit{a contradiction}, i.e., a statement that is always false).\\

\section{Burnside's theorem}

\noindent Let $\Sigma$ be the collection of all the projection operators on the finite-dimensional Hilbert space and let the dimension of this space be greater than 1.\\

\noindent Consider the set of the invariant subspaces invariant under each projection operator $\hat{P}_\diamond$ from the collection $\Sigma$:\smallskip

\begin{equation}  
   \mathcal{L}(\Sigma)
   =
   \bigcap_{\hat{P}_\diamond \in \Sigma}
      \mathcal{L}(\hat{P}_\diamond)
   \;\;\;\;  .
\end{equation}
\smallskip

\noindent The notation for this set refers to the intersections such as\smallskip

\begin{equation}  
   \mathcal{L}(\hat{P}_i^{(Q)})
   \cap
   \mathcal{L}(\hat{P}_j^{(Q^\prime)})
   =
   \mathrm{ran}(\hat{R_i})
   \cap
   \mathrm{ran}(\hat{R_j^\prime})
   =
   \mathrm{ran}(\hat{R_i}\hat{R_j^\prime})
   \;\;\;\;  ,
\end{equation}
\smallskip

\noindent where $\hat{R_i}$ and $\hat{R_j^\prime}$ stand for the sets $\{\hat{0}, \hat{P}_i^{(Q)}, \neg\hat{P}_i^{(Q)}, \hat{1}\}$ and $\{\hat{0}, \hat{P}_j^{(Q^\prime)}, \neg\hat{P}_j^{(Q^\prime)}, \hat{1}\}$, respectively, in which $\hat{P}_i^{(Q)}$ and $\hat{P}_j^{(Q^\prime)}$ are the incommutable projection operators from the different maximal contexts $\Sigma^{(Q)}=\{\hat{P}_i^{(Q)}\}_{i=1}$ and $\Sigma^{(Q^\prime)}=\{\hat{P}_j^{(Q^\prime)}\}_{j=1}$.\\

\noindent Hence, one may expect that $\mathcal{L}(\Sigma)$ contains $\mathrm{ran}(\hat{P}_i^{(Q)}) \cap \mathcal{H}$ and $\mathcal{H} \cap\, \mathrm{ran}(\hat{P}_j^{(Q^\prime)})$ as well as $\mathrm{ran}(\hat{P}_i^{(Q)}) \,\cap\, \mathrm{ran}(\hat{P}_j^{(Q^\prime)})$, i.e., the incompatible subspaces $\mathrm{ran}(\hat{P}_i^{(Q)})$ and $\mathrm{ran}(\hat{P}_j^{(Q^\prime)})$ together with their meet $\{0\}$.

\vspace*{3mm}

\begin{theorem} 
The column spaces of the incommutable projection operators cannot meet each other on the finite-dimensional Hilbert space $\mathcal{H}$ with $\mathrm{dim}(\mathcal{H}) > 1$.
\end{theorem}

\begin{proof}
Let $L(\mathcal{H})$ denote the algebra of linear transformations on $\mathcal{H}$. Since $\Sigma$ is the collection of all the projection operators on $\mathcal{H}$, $\Sigma = L(\mathcal{H})$. In the said case, according to Burnside's theorem on incommutable algebras \cite{Burnside, Rosenthal, Lomonosov, Shapiro}, \textit{the set $\mathcal{L}(\Sigma)$ is irreducible}, that is, it has no nontrivial invariant subspace. In symbols:\smallskip

\begin{equation}  
   \Sigma
   =
   L(\mathcal{H})
   \;
   \implies
   \;
   \mathcal{L}(\Sigma)
   =
   \left\{
      \mathrm{ran}(\hat{0})
      ,
      \mathrm{ran}(\hat{1})
   \right\}
   \;\;\;\;  .
\end{equation}
\smallskip

\noindent Being neither $\{0\}$ nor $\mathcal{H}$, the incompatible subspaces $\mathrm{ran}(\hat{P}_i^{(Q)})$ and $\mathrm{ran}(\hat{P}_j^{(Q^\prime)})$ are not elements of the set $\mathcal{L}(\Sigma)$, and so they cannot meet each other.
\end{proof}

\vspace*{2mm}

\begin{corollary} 
$\wedge$ cannot be a binary operation on a pair of the incommutable projection operators on the finite-dimensional Hilbert space $\mathcal{H}$ with $\dim(\mathcal{H}) > 1$.
\end{corollary}

\begin{proof}
According to (\ref{six}), $\wedge$ is a binary operation on the set of the projection operators generated by the maximal context $\mathcal{L}(\Sigma^{(Q)})$. However, since the partially ordered sets generated by the different maximal contexts, such as $\mathcal{L}(\Sigma^{(Q)})$ and $\mathcal{L}(\Sigma^{(Q^\prime)})$, do not form the common partially ordered set containing the incommutable projection operators when $\mathcal{H}$ is finite and $\dim(\mathcal{H}) > 1$, $\wedge$ cannot be defined as a binary operation on the said set. This means that if $\mathcal{H}$ is finite, $\wedge$ cannot be defined for a pair of the incommutable projection operators. 
\end{proof}

\vspace*{2mm}

\section{Concluding remarks}

\noindent One may ask, is there any mathematical fact opposing the assumption that the collection of all the column spaces containing the compatible and incompatible subspaces forms a lattice? Clearly, if such a fact is absent, then to not allow the distributive axiom in quantum theory can be regarded as justifiable, at least from a mathematical point of view.\\

\noindent As it has been demonstrated in the presented paper, the said mathematical fact is Burnside's theorem on matrix algebras. According to this theorem, if the Hilbert space $\mathcal{H}$ of the quantum system is finite, the column spaces of the incommutable projection operators cannot be elements of one partially ordered set.\\

\noindent Regarding the spin-half particle mentioned in the Introduction, this means the following.\\

\noindent Recall that any matrix $M_{2,2} \in \mathbb{C}^2$ can be expressed as

\begin{equation}  
   M_{2,2}
   =
   c\hat{1}
   +\!\!\!\!
   \sum_{Q \in \{z,x,y\}}\!\!\!\!
    a_Q\!
      \left(
         \hat{P}_{1}^{(Q)} - \hat{P}_{2}^{(Q)}
      \right)
   \;\;\;\;  ,
\end{equation}
\smallskip

\noindent where $c$ is a complex number and $a_Q$ is a 3-component complex vector. For that reason, the collection $\Sigma$ of the maximal contexts, namely,\smallskip

\begin{equation}  
   \Sigma
   =
   \{ \Sigma^{(Q)} \}_Q
   =
    \{ \hat{P}_{1}^{(Q)}, \hat{P}_{2}^{(Q)} \}_Q
   \;\;\;\;  ,
\end{equation}
\smallskip

\noindent spans the full algebra $L(\mathbb{C}^2)$ of $2\times 2$ complex matrices, i.e., $L(\mathbb{C}^2) = \Sigma$. Thus, in accordance with Burnside’s theorem, there is no subspace of $\mathbb{C}^2$ – other than $\{0\}$ and $\mathbb{C}^2$ – that every member of $\Sigma$ maps into itself. In other words, the intersection of the invariant-subspace lattices of the maximal contexts $\mathcal{L}(\Sigma^{(z)})$, $\mathcal{L}(\Sigma^{(x)})$ and $\mathcal{L}(\Sigma^{(y)})$, explicitly,\smallskip

\begin{equation}  
   \mathcal{L}(\Sigma^{(z)})
   =
   \bigg\{
      \{0\}
      ,
      \left\{
         \!\left[
            \begingroup\SmallColSep
            \begin{array}{r}
               a \\
               0
            \end{array}
            \endgroup
         \right]\!
      \right\}
      ,
      \left\{
         \!\left[
            \begingroup\SmallColSep
            \begin{array}{r}
               0 \\
               a
            \end{array}
            \endgroup
         \right]\!
      \right\}
      ,
      \mathbb{C}^2
   \bigg\}
   \;\;\;\;  ,
\end{equation}

\begin{equation}  
   \mathcal{L}(\Sigma^{(x)})
   =
   \bigg\{
      \{0\}
      ,
      \left\{
         \!\left[
            \begingroup\SmallColSep
            \begin{array}{r}
               a \\
               a
            \end{array}
            \endgroup
         \right]\!
      \right\}
      ,
      \left\{
         \!\left[
            \begingroup\SmallColSep
            \begin{array}{r}
               a \\
              -a
            \end{array}
            \endgroup
         \right]\!
      \right\}
      ,
      \mathbb{C}^2
   \bigg\}
   \;\;\;\;  ,
\end{equation}

\begin{equation}  
   \mathcal{L}(\Sigma^{(y)})
   =
   \bigg\{
      \{0\}
      ,
      \left\{
         \!\left[
            \begingroup\SmallColSep
            \begin{array}{r}
               ia \\
               a
            \end{array}
            \endgroup
         \right]\!
      \right\}
      ,
      \left\{
         \!\left[
            \begingroup\SmallColSep
            \begin{array}{r}
               a \\
              ia
            \end{array}
            \endgroup
         \right]\!
      \right\}
      ,
      \mathbb{C}^2
   \bigg\}
   \;\;\;\;  ,
\end{equation}
\smallskip

\noindent is irreducible, that is,\smallskip

\begin{equation}  
   \mathcal{L}(\Sigma)
   =
   \bigcap_{Q \in \{z,x,y\}}
      \mathcal{L}(\Sigma^{(Q)})
   =
   \left\{
      \{0\}
      ,
      \mathbb{C}^2
   \right\}   
   \;\;\;\;  .
\end{equation}
\smallskip

\noindent Hence, the meets such as $\mathrm{ran}(\hat{P}_{1}^{(z)}) \wedge \mathrm{ran}(\hat{P}_{1}^{(x)}) = \{[\begin{smallmatrix} a \\ 0 \end{smallmatrix}]\} \cap \{[\begin{smallmatrix} a \\ a \end{smallmatrix}]\} = \{0\}$ have no meaning as $ \{[\begin{smallmatrix} a \\ 0 \end{smallmatrix}]\}, \{[\begin{smallmatrix} a \\ a \end{smallmatrix}]\} \notin \mathcal{L}(\Sigma)$. Accordingly, the operation $\wedge$ cannot be defined for a pair of the incommutable projection operators such as $\hat{P}_1^{(z)}$ and $\hat{P}_1^{(x)}$.\\

\bibliographystyle{References}
\bibliography{Meet_Ref}

\end{document}